\documentclass[pra,groupedaddress,superscriptaddress,twocolumn,showpacs,amsmath,amssymb,a4paper]{revtex4}
\usepackage{geometry}
\usepackage{graphicx}
\usepackage{mathrsfs}
\usepackage{amssymb}
\usepackage{amsmath}
\usepackage{amsthm}
\usepackage{bm}
\usepackage{hyperref}

\newtheorem{proposition}{Proposition}
\newtheorem{theorem}{Theorem}

\newtheorem{corollary}{Corollary}

\theoremstyle{definition}
\newtheorem{definition}{Definition}

\newtheorem{example}{Example}
\newtheorem{remark}{Remark}

\newcommand{\ave}[1]{\langle #1 \rangle}
\newcommand{\bra}[1]{\langle #1|}
\newcommand{\ket}[1]{| #1 \rangle }

\newcommand{\be}{\begin{eqnarray}}
\newcommand{\ee}{\end{eqnarray}}

\begin{document}

\title{Spectral properties of reduced fermionic density operators \\ and parity superselection rule}

\author{Grigori G. Amosov}

\email{gramos@mi.ras.ru}

\affiliation{V.~A.~Steklov Mathematical Institute, Russian Academy
of Sciences, Gubkina Str. 8, Moscow, 119991, Russia}

\author{Sergey N. Filippov}

\email{sergey.filippov@phystech.edu}

\affiliation{Institute of Physics and Technology, Russian Academy
of Sciences, Nakhimovskii Pr. 34, Moscow 117218, Russia}

\affiliation{Moscow Institute of Physics and Technology,
Institutskii Per. 9, Dolgoprudny, Moscow Region 141700, Russia}

\affiliation{P.~N.~Lebedev Physical Institute, Russian Academy of
Sciences, Leninskii Pr. 53, Moscow 119991, Russia}

\begin{abstract}
We consider pure fermionic states with a varying number of
quasiparticles and analyze two types of reduced density operators:
one is obtained via tracing out modes, the other is obtained via
tracing out particles. We demonstrate that spectra of mode-reduced
states are not identical in general and fully characterize pure
states with equispectral mode-reduced states. Such states are
related via local unitary operations with states satisfying the
parity superselection rule. Thus, valid purifications for
fermionic density operators are found. To get particle-reduced
operators for a general system, we introduce the operation
$\Phi(\varrho) = \sum_i a_i \varrho a_i^{\dag}$. We conjecture
that spectra of $\Phi^p(\varrho)$ and conventional $p$-particle
reduced density matrix $\varrho_p$ coincide. Nontrivial
generalized Pauli constraints are derived for states satisfying
the parity superselection rule.
\end{abstract}



%03.65.Aa    Quantum systems with finite Hilbert space
%03.65.Fd    Algebraic methods
%05.30.Fk    Fermion systems and electron gas
%03.67.Bg    Entanglement production and manipulation

\maketitle

%%%%%%%%%%%%%%%%%%%%%%%%%%%%%%%%%%%%%%%%%%%%%%%%%%%%%%%%%%%%%%%%%%%%%%%%%%%%%%%%%%%%

\section{\label{section-introduction} Introduction}

The physical nature of information~\cite{landauer} not only
provides new ways of information processing (like quantum
computing~\cite{nielsen-chuang}) but also stimulates the research
of specific information carriers --- particles and quasiparticles
with peculiar properties. This paper deals with particular quantum
states which represent a superposition of states with different
numbers of fermionic quasiparticles.

After electrons were argued to possess the additional degree of
freedom~\cite{pauli-1924} (later called
spin~\cite{uhlenbeck-1925}) and obey the Pauli exclusion
principle~\cite{pauli-1925}, the Fermi--Dirac statistics was
found~\cite{fermi,dirac} and the general notion of fermion
particles emerged. The antisymmetric nature of fermionic
wavefunction~\cite{dirac} has provoked development of the
canonical anticommutation relation~\cite{jordan-wigner} and the
general theory of second quantization~\cite{fock}, which deals
with systems with a varying number of particles. The fermionic
statistics of half-integer spin particles was established
in~\cite{fierz,pauli-1940}. Since every fermion has a half-integer
spin, the total spin of an even number of fermions must be
integer, whereas that of an odd number of fermions must be
half-integer. Given a system with a varying number of particles,
it is therefore possible to have a superposition of half-integer
and integer total spins. However, it was recognized later that the
comparison of phases between states with half-integer and integer
angular momenta cannot be reconciled with the requirement of
relativistic invariance~\cite{www}, which resulted in the
formulation of parity (spinor, univalence) superselection rule
(see, e.g., the reviews~\cite{whiteman,cisneros,bartlett}). Any
superselection rule relies on a group of physical transformations,
for example the parity superselection rule relies on the
rotational invariance~\cite{hegerfeldt} whereas the mass
superselection rule relies on the Galilean invariance of
non-relativistic quantum mechanics~\cite{bargman}. If a system of
leptons were invariant with respect to fermionic $U(1)$
transformations, there would be a leptonic family number
superselection rule, however, the recent experiments clearly
indicate the neutrino oscillations (see, e.g.,
\cite{gonzalez-garcia}), which contradicts to the conservation of
leptonic family number (the violation was predicted in
Ref.~\cite{pontecorvo}). Thus, the experimental investigations of
nature modify our understanding of fundamental symmetries, the
physical models, and the mathematical frameworks used for their
description. Recent experiments with massless Weyl fermionic
quasiparticles~\cite{xu-science,lu,xu-nat-phys} stimulate us to
investigate properties of fermionic states in coherent
superpositions of different number states. Moreover,
superpositions of different number states usually emerge in
fermionized models (see, e.g.~\cite{tsvelik}).

In this paper, we analyze properties of general pure fermionic
states with a varying number of quasiparticles from the viewpoint
of quantum information theory, namely, we consider the spectra of
two reduced operators obtained from such a state by two different
approaches: reduction of modes and reduction of particles.

Suppose a state space $H$ composed of two sets of modes $H_1$ and
$H_2$. The mode represents a single-particle state that is either
occupied or not. In quantum information theory, $H_1$ and $H_2$
may correspond to distant laboratories. In solid state physics,
$H_1$ and $H_2$ may describe different positions or momenta of
quasiparticles. In quantum chemistry, $H_1$ and $H_2$ may be
attributed to two distinct sets of spin-orbits. In this paper, we
consider situations, when the number of particles may change not
only in $H_1$ and $H_2$ but in the common space $H$ too. Suppose a
local physical observable $A_1 \in \mathcal{A}_1$ acting on modes
$H_1$ only, for instance, $A_1$ can be the local energy or the
population of modes in $H_1$. On the one side, the
quantum-mechanical mean value $\ave{A_1} = \omega(A_1) = {\rm
Tr}(\varrho A_1)$, where $\omega$ is a functional defined on the
common algebra of operators $\mathcal{A} = \mathcal{A}_1 \times
\mathcal{A}_2$ and $\varrho$ is the total density operator. On the
other side, $\ave{A_1} = \omega_1(A_1) = {\rm Tr}(\varrho_1 A_1)$,
where $\omega_1$ is a functional defined on the algebra of local
operators $\mathcal{A}_1$ and $\varrho_1$ is the mode-reduced
density operator. \footnote{It is worth mentioning, that
particular realizations of $\varrho_1$ are known in quantum
chemistry as orbital reduced density matrices. In this case, a
single orbital corresponds to the algebra $\mathcal{A}_1$
generated by operators $a_{\uparrow}a_{\uparrow}^{\dag}$,
$a_{\uparrow}$, $a_{\uparrow}^{\dag}$,
$a_{\uparrow}^{\dag}a_{\uparrow}$ and
$a_{\downarrow}a_{\downarrow}^{\dag}$, $a_{\downarrow}$,
$a_{\downarrow}^{\dag}$, $a_{\downarrow}^{\dag}a_{\downarrow}$
corresponding to electron spin projections $+\frac{1}{2}$ and
$-\frac{1}{2}$, respectively. In our notation, the 1-orbital
reduced density matrix $\varrho_1$ is a two-mode partial state as
it takes into account both spin projections. Analogously,
2-orbital reduced density matrix is a four-mode partial state.}
Analogously, $\omega_2$ and $\varrho_2$ represent a mode-reduced
state of the second set of modes. Thus, for physical systems with
two separate sets of single-particle sets, mode reduced states
naturally occur as algebraic constructions. Such algebraic
constructions are well defined for systems with a variable number
of quasiparticles, so we explore spectral properties of $\omega_1$
and $\omega_2$ in this paper.

The inverse procedure to reduction is known as a purification and
takes a simple form in tensor product spaces because the two
spectra are identical (an immediate consequence of the Schmidt
decomposition). We demonstrate that spectra of mode-reduced states
do not necessarily coincide for a general fermionic state, which
can be interpreted as a violation~\cite{moriya-2002} of the
Araki--Lieb inequality~\cite{araki-lieb}. This fact should be
taken into account when the entanglement of pure states is
quantified via the entropy of one of subsystems~\cite{Friis}. The
discrepancy between spectra (entropies) of subsystems of a pure
state seems to be unphysical (since the result depends on
labelling), so we further find necessary and sufficient conditions
for spectra to be coincident for mode-reduced states. In
particular, we not only find that spectra coincide for states
satisfying the parity superselection rule~\cite{friis-2015} but
also prove that all states with equispectral subsystems are the
univalent states subjected to local unitary transformations.

Suppose now that one is interested in an average $p$-particle
property, for instance, the average spin projection ($p=1$) or the
average two-electron Coulomb interaction ($p=2$). In this case,
one should trace out all irrelevant particles. In quantum
chemistry, a particle reduction is realized by integrating over
irrelevant particles' coordinates~\cite{carlson,coleman,ando},
which results in the so-called $p$-particle reduced density matrix
(see, e.g., \cite{mazziotti}). Such an approach provides new
inequalities for fermionic occupation numbers and leads to the
generalized Pauli
constraints~\cite{borland-dennis,Schilling,benavides-riveros,schilling-2015}.
As far as systems with a varying number of particles are
concerned, the construction of $p$-particle reduced density
matrices is non-trivial and we discuss it in this paper. Moreover,
we also consider the traced out part $\Phi^p(\varrho)$ with a
varying number of particles. Thus, we generalize
particle-reduction technique for systems with a varying number of
quasiparticles and compare the spectra of $p$-particle reduced
density matrix and $\Phi^p(\varrho)$.

As a by-product, we formalize the density matrix construction for
states with a varying number of fermions. Our description of
fermionic density operator differs from that used by Cahill and
Glauber~\cite{cahill-glauber-1999}.

The problem of our interest --- possible pure fermionic states
with equispectral reduced density operators --- is connected with
the problem of calculating von Neumann entropy for the fermionic
dynamical systems and estimating capacities of fermionic quantum
channels~\cite{Bravyi,Fannes}. One may expect that the transfer of
quantum information through a fermionic channel would be different
from that for bosons~\cite{lewenstein}, as the fermionic theory
differs from the qubit theory~\cite{dariano-ijmp-2014,zimboras} in
both tomography~\cite{dariano-2014} and quantum
computation~\cite{bravyi-kitaev}.

The paper is organized as follows. In Sec.~\ref{section-notation},
we describe the notation used and formulate the problem. In
Sec.~\ref{section-density}, the density matrix formalism is
concisely clarified. In Sec.~\ref{section-spectra}, the main
results regarding the spectra of mode-reduced states are obtained.
In Sec.~\ref{section-particle-reduction}, the construction and
spectrum of particle-reduced operators are analyzed. In
Sec.~\ref{section-conclusions}, brief conclusions are given.

%%%%%%%%%%%%%%%%%%%%%%%%%%%%%%%%%%%%%%%%%%%%%%%%%%%%%%%%%%%%%%%%%%%%%%%%%%%%%%%%%%%%%%%%%%%%%%%%%%%%

\section{\label{section-notation} Notation}

Consider the algebra of canonical anticommutation relations (CAR
algebra) $\mathcal{A}$  generated by the ladder operators $a_s$,
$a_s^{\dag}$ satisfying the relations~\cite{jordan-wigner,berezin}
\begin{eqnarray}
&& a_s a_t^{\dag} + a_t^{\dag} a_s = \delta _{st} I, \nonumber\\
&& a_s a_t + a_t a_s = a_s^{\dag} a_t^{\dag} + a_t^{\dag}
a_s^{\dag} = 0, \nonumber
\end{eqnarray}

\noindent where $s,t = 1, \ldots, n$. Let $H_n$ be a Hilbert space
with the dimension ${\rm dim}H_n=2^n$ and $\{ \ket{j_1 \ldots j_n}
\}$ be a fixed orthonormal basis, where indices $j_s= 0,1$. We
shall suppose that $\mathcal{A}$ is realized as the algebra of
operators in $H_n$ such that
\begin{eqnarray}
&& \!\!\!\!\! a_t^{\dag} \ket{j_1 \ldots j_n} \nonumber\\
&& \!\!\!\!\! = \left \{
\begin{array}{l} \! \exp\left( i \pi \sum\limits_{s=1}^{t-1}j_s \right)
\ket{j_1\ldots j_{t-1} \, 1 \,  j_{t+1}\ldots j_n} \text{~if~} j_t=0,\\
\! 0 \quad {\rm if} \quad j_t=1,\end{array}\right. \quad
\end{eqnarray}
and $a_t$ is conjugate to $a_t^{\dag}$.

A positive functional $\omega \in B(H_n)^{\ast}$ normed by the
condition $\omega (I) = 1$ is called a state on the algebra of all
bounded operators $B(H_n)$ in $H_n$. Given a state $\omega$, there
exists a unique positive unit-trace operator $\varrho $ such that
$\omega(x) = {\rm Tr}(\varrho x),\ x\in B(H_n)$. {\it A spectrum
of $\omega$} is defined as the spectrum of $\varrho$. The
following definition~\cite{wiseman-vaccaro} will play a major role
in our analysis:

\begin{definition} The state $\ket{\psi}$ satisfies the parity
superselection rule if a unit vector $\ket{\psi} \in H$ has the
form
\begin{equation}\label{1}
\ket{\psi} =\sum \limits _{j_1,\ldots ,j_n\in \{0,1\}}\lambda
_{j_1\ldots j_n} \ket{j_1\ldots j_n},\ \lambda _{j_1\ldots j_n}\in
{\mathbb C},
\end{equation}
where all the numbers $\sum_s j_s$ are even or odd alternatively
for all non-zero $\lambda _{j_1\ldots j_n}$.
\end{definition}

The set of all states satisfying the superselection rule contains
the important subset consisting of pure states for which the
number of particles $N$ is fixed. If this is the case, non-zero
terms in \eqref{1} satisfy the condition $\sum_s j_s = N = {\rm
const}$. It is worth mentioning that all pure quasifree fermionic
states have a fixed number of particles~\cite{evans,dierckx}.

Another important class of fermionic states is determined by their
peculiar action on monomials of the odd order:

\begin{definition}
The state $\omega$ is said to be even if
\[
\omega (a^{\#}_{s_1}\ldots a^{\#}_{s_{2k+1}})=0
\]
for any choice of $2k+1$ ladder operators $a^{\#}_s=a_s$ or
$a^{\#}_s=a_s^{\dag}$.
\end{definition}

The relation between the above classes becomes especially clear in
the case of pure states.

\begin{proposition}
\label{prop-even} The pure state $\omega$ is even iff it satisfies
the parity superselection rule.
\end{proposition}
\begin{proof}
Suppose that $\omega$ satisfies the parity superselection rule.
Consider the representation \eqref{1}. Since the odd order
monomial $a^{\#}_{s_1} \ldots a^{\#}_{s_{2k+1}}$ changes even
number of particles to the odd one and vice versa, the vectors
$a^{\#}_{s_1} \ldots a^{\#}_{s_{2k+1}} \ket{j_1\ldots j_n}$ and
$\ket{\tilde{j}_1\ldots \tilde{j}_n}$ are orthogonal for any
choice of $a^{\#}_s$ if $\ket{j_1\ldots j_n}$ and
$\ket{\tilde{j}_1\ldots \tilde{j}_n}$ correspond to non-zero terms
in Eq.~\eqref{1}. It follows that $\omega$ is even.

If $\omega$ does not satisfy the parity superselection rule, then
its representation in the form of \eqref{1} contains at least two
non-zero terms, say $\ket{j_1\ldots j_n}$ and $\ket{\tilde {j}_1
\ldots \tilde{j}_n}$, such that $\sum_s j_s$ and $\sum_s
\tilde{j}_s$ have different parity. There exists a unique monomial
of the odd order $a_{s_1} \ldots a_{s_{2k+1}}$, which is a partial
isometrical operator mapping $\ket{j_1 \ldots j_n}$ to
$\ket{\tilde{j}_1 \ldots \tilde{j}_n}$ and mapping all other basis
vectors to zero. Thus, the value of $\omega$ is not equal to zero
on this monomial. Hence, $\omega$ cannot be even.
\end{proof}

Before we proceed to the analysis of mode and particle reductions
of fermionic states, we construct the density matrix formalism in
the convenient form.

%%%%%%%%%%%%%%%%%%%%%%%%%%%%%%%%%%%%%%%%%%%%%%%%%%%%%%%%%%%%%%%%%%%%%%%%%%%%%%%%%%%%%%%%%%%%%%%%%

\section{\label{section-density} Density matrix for fermions}

Let us consider $n$ fermionic modes and the corresponding CAR
algebra $\mathcal{A}$, ${\rm dim}{\mathcal A}=4^n$. The generators
of $\mathcal{A}$ are $a_s a_s^{\dag}$, $a_s$, $a_s^{\dag}$, and
$a_s^{\dag} a_s$, $s = 1,\ldots,n$. Binary representation of
numbers $0,\ldots,2^n -1$ forms the set of multiindices ${\mathcal
J}\ni J=(j_1,\ldots ,j_n),\ j_s\in \{0,1\}$. Given two
multiindices $J,K \in \mathcal{J}$, let us define $A_{JK}\in
\mathcal{A}$ by the formula
\begin{equation}
\label{araki-wyss-formula} A_{JK} = c_{j_1}^{\dag} c_{j_2}^{\dag}
\ldots c_{j_n}^{\dag} c_{k_n} \ldots c_{k_2} c_{k_1},
\end{equation}
\noindent where
\begin{equation}
\label{c-operators} c_{j_s}^{\dag} = \left \{\begin{array}{ll} a_s a_s^{\dag} & \ {\rm if} \quad j_s=0,\\
a_s^{\dag} & \ {\rm if} \quad j_s=1;\end{array}\right. \quad
c_{k_s} = \left \{\begin{array}{ll} a_s a_s^{\dag} & \ {\rm if} \quad k_s=0,\\
a_s & \ {\rm if} \quad k_s=1.\end{array}\right.
\end{equation}

\begin{proposition} \label{prop-matrix-unit} The following relation holds:
\[
A_{JK}= \ket{j_1\ldots j_n} \bra{k_1\ldots k_n}.
\]
\end{proposition}
\begin{proof}
Consider the operator $C_K = c_{k_n}\ldots c_{k_1}$, which is a
rank one partial isometry in the sense that
$C_K=\ket{\varphi}\bra{\chi}$ for some unit vectors
$\ket{\varphi}, \ket{\chi} \in H$ that are either orthogonal or
coincide. Moreover, $C_K \ket{k_1 \ldots k_n} = \ket{0\ldots 0}$.
Then, $A_{JK} = C_J^{\dag} C_K$, which concludes the proof.
\end{proof}

The construction similar to that in Eq.~\eqref{araki-wyss-formula}
was used in the paper~\cite{Araki} for other purposes.

\begin{corollary}
\label{corollary-1} The element $\varrho \in {\mathcal A}$ defines
a valid quantum state iff it can be represented in the form
\begin{equation}\label{star}
\varrho =\sum\limits_{J,K\in {\mathcal J}} \lambda_{JK} A_{JK},
\end{equation}
where $(\lambda_{JK})$ is a positive semidefinite matrix with the
unit trace $\sum_J\lambda_{JJ}=1$ (fermionic density matrix).
\end{corollary}

\begin{proof}
Any element of the algebra $\mathcal{A}$ can be represented in the
form \eqref {star} because $(A_{JK})$ are matrix units due to
Proposition~\ref{prop-matrix-unit}. The functional $\omega(x) =
\sum_{J,K\in \mathcal{J}} \lambda_{JK} {\rm Tr}(A_{JK} x)$ gives
non-negative values for positive semidefinite operators $x$ if and
only if $(\lambda_{JK})$ is a positive semidefinite matrix too (as
the functional is a trace of the product of two matrices).
Finally, $\omega(I) = \sum_J \lambda_{JJ} = 1$.
\end{proof}

The $2^n \times 2^n$ matrix $(\lambda_{JK})$ is a fermionic
density matrix and fulfils the conventional requirements
(Hermitian, positive semidefinite, and unit trace). The following
Corollary provides a convenient description of the density
operator.

\begin{corollary}
\label{corollary-matrix-of-averages} The $2^n \times 2^n$ matrix
with elements $\ave{A_{KJ}}$ is a density matrix of an $n$-mode
fermionic state over which the average is taken.
\end{corollary}
\begin{proof}
In fact, $\ave{A_{KJ}} = {\rm Tr}(A_{JK}^{\dag} \varrho) =
\lambda_{JK}$.
\end{proof}

\begin{example}
\label{example-density-matrices} Using
Eq.~\eqref{araki-wyss-formula} and
Corollary~\ref{corollary-matrix-of-averages}, the 2-mode fermionic
state can be described by the (total) density matrix
\begin{equation}
\label{total-density-matrix} \left(
\begin{array}{cccc}
  \ave{a_1 a_1^{\dag} a_2 a_2^{\dag}} & \ave{a_1 a_1^{\dag} a_2^{\dag}} & \ave{a_1^{\dag} a_2 a_2^{\dag}} & \ave{a_1^{\dag} a_2^{\dag}} \\
  \ave{a_1 a_1^{\dag} a_2} & \ave{a_1 a_1^{\dag} a_2^{\dag} a_2} & \ave{a_1^{\dag} a_2} & \ave{a_1^{\dag} a_2^{\dag} a_2} \\
  \ave{a_1 a_2 a_2^{\dag}} & \ave{a_2^{\dag} a_1} & \ave{a_1^{\dag} a_1 a_2 a_2^{\dag}} & \ave{a_1^{\dag} a_2^{\dag} a_1} \\
  \ave{a_2 a_1} & \ave{a_1 a_2^{\dag} a_2} & \ave{a_1^{\dag} a_2 a_1} & \ave{a_1^{\dag} a_1 a_2^{\dag} a_2} \\
\end{array}
\right),
\end{equation}

\noindent which is constructed with the help of conventional sets
of multiindices $\mathcal{J},\mathcal{K} =
\{(0,0);(0,1);(1,0);(1,1)\}$ and differs from symbols used in the
description of fermionic Gaussian states~\cite{bravyi-2005}
(reviewed, e.g., in~\cite{greplova}). \hfill $\square$
\end{example}

%%%%%%%%%%%%%%%%%%%%%%%%%%%%%%%%%%%%%%%%%%%%%%%%%%%%%%%%%%%%%%%%%%%%%%%%%%%%%%%%%%%%

\section{\label{section-spectra} Spectra of mode-reduced states}

In this section, we deal with coherent superpositions of different
number states, so we exploit the bipartition based not on the
particles but rather on fermionic modes. Fermionic modes can be
thought of as nodes that can be either occupied or not by
particles. The bipartition is merely an aggregation of all nodes
into two groups. For instance, bipartitions with respect to
different spin components were considered recently~\cite{paraan}.
Mathematically, we deal with the algebraic bipartition in the
second quantized
description~\cite{zanardi,amosov-mancini,benatti-osid-2014,benatti-2014,marzolino},
which is used in the study of
entanglement~\cite{wiseman-vaccaro,shi,banuls,iemini,gigena}.

Partition of modes is performed among single-particle states that
fermions can potentially occupy, for example, nodes of some
lattice potential in coordinate configuration space, or states
with different momentum in momentum configuration space. Let the
first subsystem contain $m$ modes and the second one contain $n-m$
modes. In mathematical terms, $\mathcal{A}$ is obtained by joining
two algebras of canonical anticommutation relations ${\mathcal
A}_1$ and ${\mathcal A}_2$ with the generators $\{a_1,
a_1^{\dag},\ldots, a_m, a_m^{\dag}\}$ and
$\{a_{m+1},a_{m+1}^{\dag}$, $\ldots, a_n, a_n^{\dag}\}$,
respectively. Fix a unit vector $\ket{\psi} \in H$.

\begin{definition} The states $\omega_j(x)= \bra{\psi} x \ket{\psi},\ x\in {\mathcal A}_j$
that are the restrictions of a pure state $\ket{\psi} \bra{\psi}$
to the algebras ${\mathcal A}_j$, $j=1,2$, are said to be
mode-reduced (partial) states of $\ket{\psi} \bra{\psi}$.
\end{definition}

\begin{example}
For a state (\ref {total-density-matrix}) the reduced density matrices describing the first mode and the
second mode read
\begin{equation}
\label{reduced-matrices} \left(
\begin{array}{cc}
  \ave{a_1 a_1^{\dag}} & \ave{a_1^{\dag}} \\
  \ave{a_1} & \ave{a_1^{\dag} a_1} \\
\end{array}
\right) \quad \text{and} \quad \left(
\begin{array}{cc}
  \ave{a_2 a_2^{\dag}} & \ave{a_2^{\dag}} \\
  \ave{a_2} & \ave{a_2^{\dag} a_2} \\
\end{array}
\right),
\end{equation}

\noindent respectively. Note that both matrices
\eqref{reduced-matrices} cannot be simultaneously obtained from
the total density matrix \eqref{total-density-matrix} by
conventional partial trace methods. This is a peculiar property of
fermionic states which differs them from the bosonic ones. \hfill
$\square$

\end{example}

\begin{remark}
Let us clarify the relation between the introduced definition of
the mode-reduced states and the orbital reduced density matrices
used in quantum chemistry \cite{boguslawski-2013,rissler-2006}.

In quantum chemistry, the mode is a composite of the space orbital
state and the electron spin projection ($\ket{\uparrow}$ or
$\ket{\downarrow}$), and the whole system state vector is written
in the form $\ket{\Psi} = \sum_{n_1,\ldots,n_L}
\psi_{n_1,\ldots,n_L} \ket{n_1} \otimes \ldots \otimes \ket{n_L}$,
where $n_i = 0,1$ indicates the occupation of the corresponding
spin-orbit, $L$ is the number of active spin-orbits. The
antisymmetric property of the electron wavefunction is encoded in
the coefficients $\psi_{n_1,\ldots,n_L}$, and the global basis
states $\ket{n_1} \otimes \ldots \otimes \ket{n_L}$ already have
the tensor product structure. The reduced states are then readily
obtained by tracing out undesired orbits.

Alternatively, one can use antisymmetric vectors (Slater
determinants) $\ket{e_1} = \ket{n_1 \ldots n_{i-1}}$, $\ket{e_2} =
\ket{n_{i+1} \ldots n_{j-1}}$, $\ket{e_3} = \ket{n_{j+1} \ldots
n_L}$ and expand $\ket{\Psi} = \sum_{n_i,n_j,e_1,e_2,e_3}
\psi_{n_i,n_j,e_1,e_2,e_3} \ket{e_1} \otimes \ket{n_i} \otimes
\ket{e_2} \otimes \ket{n_j} \otimes \ket{e_3}$. The reduced 2-mode
state reads ${\rm Tr}_{e_1,e_2,e_3} \ket{\Psi}\bra{\Psi}$ as each
vector $\ket{e_1} \otimes \ket{n_i} \otimes \ket{e_2} \otimes
\ket{n_j} \otimes \ket{e_3}$ does have a tensor product structure.
Using the orbit basis $\ket{n_i} =
\{\ket{0},\ket{\uparrow},\ket{\downarrow},\ket{\uparrow\downarrow}\}$,
one gets the 2-orbital reduced density
operator~\cite{boguslawski-2013}.

In our approach, the global basis states $\ket{j_1 \ldots j_n}$
are antisymmetric automatically and coefficients $\lambda
_{j_1\ldots j_n}$ in Eq.~\eqref{1} are arbitrary. Thus, we deal
with basis states that do not have the tensor product structure at
all. Similarly, the action of mode operators $a_t$ and
$a_t^{\dag}$ is non-local. Since transitions between the states
with different numbers of particles are prohibited in quantum
chemistry \cite{boguslawski-2013,rissler-2006}, both
quantum-chemical and our definitions of reduced states do coincide
for systems with a fixed number of electrons. The downside of the
quantum-chemical definition may only appear for systems with
variable number of particles, since the creation and annihilation
operators cannot be represented in the tensor product form $I
\otimes \ldots \otimes I \otimes \ket{1}\bra{0} \otimes I \otimes
\ldots \otimes I$ or $I \otimes \ldots \otimes I \otimes
\ket{0}\bra{1} \otimes I \otimes \ldots \otimes I$. \hfill
$\square$
\end{remark}

Our further goal is to characterize all the states that have
equispectral mode-reduced states. To anticipate general results,
let us begin with the simplest case of two fermionic modes that
can be occupied by a system with a varying number of
quasiparticles in a pure state $\ket{\psi} = c_{00}\ket{00} +
c_{01}\ket{01} + c_{10}\ket{10} + c_{11} \ket{11}$. According to
Example~\ref{example-density-matrices} the total density matrix
\eqref{total-density-matrix} for such a state is $(c_{00}, c_{01},
c_{10}, c_{11})^{\top}$ $\times$ $(\overline{c_{00}},
\overline{c_{01}}, \overline{c_{10}}, \overline{c_{11}})$ and
corresponds to the pure state indeed, whereas the reduced density
matrices \eqref{reduced-matrices} take the form
\begin{eqnarray}
\label{reduced-matrix-1} && \Lambda_1 = \left(
\begin{array}{cc}
  |c_{00}|^2 + |c_{01}|^2 & c_{00}\overline{c_{10}} + c_{01} \overline{c_{11}} \\
  \overline{c_{00}}c_{10} + \overline{c_{01}} c_{11} & |c_{10}|^2 + |c_{11}|^2 \\
\end{array}
\right) \quad \text{and} \nonumber\\
&& \label{reduced-matrix-2} \Lambda_2 = \left(
\begin{array}{cc}
  |c_{00}|^2 + |c_{10}|^2 & c_{00}\overline{c_{01}} - c_{10} \overline{c_{11}} \\
  \overline{c_{00}}c_{01} - \overline{c_{10}} c_{11} & |c_{01}|^2 + |c_{11}|^2 \\
\end{array}
\right), \nonumber
\end{eqnarray}

\noindent respectively. The spectra of $\Lambda_1$ and $\Lambda_2$
would be identical if and only if ${\rm Tr}(\Lambda_1) = {\rm
Tr}(\Lambda_2)$ and ${\rm Tr}(\Lambda_1^2) = {\rm
Tr}(\Lambda_2^2)$. The first condition always holds true whereas
the second one reduces to ${\rm Tr}(\Lambda_1^2 - \Lambda_2^2) = 8
{\rm Re}(c_{00} c_{11} \overline{ c_{01} c_{10} }) = 0$. This
observation can be summarized as follows.

\begin{proposition}
\label{prop-two-mode-case} The two-mode fermionic state
$\ket{\psi} = c_{00}\ket{00} + c_{01}\ket{01} + c_{10}\ket{10} +
c_{11} \ket{11}$ has equispectral mode-reduced density operators
if and only if ${\rm Re}(c_{00} c_{11} \overline{ c_{01} c_{10} })
= 0$.
\end{proposition}

\begin{example}
\label{example-spectra} Suppose $\ket{\psi} =\frac {1}{2}(\ket{00}
+ \ket{01} + \ket{10} + \ket{11})$, then ${\rm
Spect}(\omega_1)=\{0,1\}$ whereas ${\rm
Spect}(\omega_2)=\{\frac{1}{2},\frac{1}{2}\}$. \hfill $\square$
\end{example}

\begin{theorem} \label{theorem-1} Suppose that a pure state $\omega$ satisfies
the parity superselection rule. Then, the spectra of $\omega_1$
and $\omega_2$ coincide.
\end{theorem}
\begin{proof}
Let us define an isometry $U:H_n\to H_{m}\otimes H_{n-m}$ by the
formula
\begin{equation}\label{u}
U \ket{j_1\ldots j_n} = \ket{j_1 \ldots j_m} \otimes \ket{j_{m+1}
\ldots j_{n}}.
\end{equation}
Then
\begin{equation}
U a_s^{\#} U^{\dag} = \left\{ \begin{array}{ll}
  a_s^{(1)\#}\otimes I & \text{if} \quad s = 1, \ldots, m, \\
  \Gamma \otimes a_s^{(2)\#} & \text{if} \quad s = m+1, \ldots, n, \\
\end{array} \right.
\end{equation}

\noindent where the fermionic operators $a_s^{(1)\#}$ and
$a_s^{(2)\#}$ act on Hilbert spaces $H_{m}$ and $H_{n-m}$,
respectively, and $\Gamma = \prod_{s=1}^m (a_s^{(1)} a_s^{(1)\dag}
- a_s^{(1)\dag} a_s^{(1)})$.

Consider the pure state $\Omega (X)=\omega (U^{\dag} X U)$, $X\in
B(H_m\otimes H_{n-m})$. The spectra of $\Omega _1 = {\rm
Tr}_{H_{n-m}}(\Omega)$ and $\Omega_2 = {\rm Tr}_{H_m}(\Omega)$ are
known to coincide as an immediate consequence of the Schmidt
decomposition in tensor product Hilbert space $H_m\otimes H_{n-m}$
(see, e.g., \cite{holevo}). It follows from the construction that
for all $x\in {\mathcal A}_1$ we have
\begin{equation}
\label{U-reduction-1} \omega_1(x) = \omega(x) = \Omega (U x
U^{\dag}) = \Omega (x^{(1)} \otimes I) = \Omega_1 (x),
\end{equation}
therefore, $\omega_1$ and $\Omega_1$ have the same spectra.

To prove that spectra of $\omega_2$ and $\Omega_2$ coincide too,
we first notice the trivial action of these functionals on odd
order monomials, namely, $\omega_2(a^{\#}_{s_1}\ldots
a^{\#}_{s_{2k+1}}) = \omega (a^{\#}_{s_1}\ldots a^{\#}_{s_{2k+1}})
= 0$ because $\omega$ is even state in virtue of
Proposition~\ref{prop-even}. On the other hand, for all $m+1 \le
s_p \le n$
\begin{eqnarray}
&& \Omega_2 (a^{(2)\#}_{s_1} \ldots a^{(2)\#}_{s_{2k+1}}) = \Omega
(I \otimes a^{(2)\#}_{s_1} \ldots a^{(2)\#}_{s_{2k+1}})
\nonumber\\
&& = \omega \left(\prod_{s=1}^m (a_s a_s^{\dag} - a_s^{\dag} a_s)
a^{\#}_{s_1}\ldots a^{\#}_{s_{2k+1}} \right) = 0
\label{equality-odd}
\end{eqnarray}

\noindent because $\omega$ is even. Thus, both $\omega_2$ and
$\Omega_2$ vanish on odd monomials. As far as even monomials are
concerned,
\begin{eqnarray}
&& \omega (a^{\#}_{s_1}\ldots a^{\#}_{s_{2k}}) = \Omega (U
a^{\#}_{s_1} \ldots a^{\#}_{s_{2k}} U^{\dag}) \nonumber\\
&& =  \Omega (\Gamma^{2k} \otimes a^{(2)\#}_{s_1} \ldots
a^{(2)\#}_{s_{2k}} ) = \Omega_2 (s^{(2)\#}_{s_1} \ldots
a^{(2)\#}_{s_{2k}} ) \quad \label{equality-even}
\end{eqnarray}

\noindent because $\Gamma^{2k} = I$. Thus, $\omega_2$ and
$\Omega_2$ coincide on even monomials too. To conclude, ${\rm
Spect}(\omega_1) = {\rm Spect}(\Omega_1) = {\rm Spect}(\Omega_2)
={\rm Spect}(\omega_2)$.
\end{proof}

The obtained result shows that states satisfying the parity
superselection rule have equispectral mode-reduced states and,
therefore, are physical. Note that superselected states do not
necessary have a fixed number of particles, and we believe that a
coherent superposition (not a mixture) of states with different
numbers of fermionic quasiparticles can be observed in future
experiments \cite{xu-science,lu,xu-nat-phys}. From a mathematical
viewpoint, it is interesting to find all possible states with
equispectral mode-reduced states. The answer to this question
provides the following theorem.

\begin{theorem} \label{theorem-2} Suppose that for a pure state $\ket{\psi}\bra{\psi}$ the partial
states $\omega _1$ and $\omega _2$ have identical spectra. Then,
there exist a state $\ket{\phi}\bra{\phi}$ satisfying the parity
superselection rule and unitary operators $U_1\in {\mathcal A}_1,\
U_2\in {\mathcal A}_2$ such that the partial states of $\ket{U_2
U_1 \phi}\bra{U_2 U_1 \phi}$ coincide with $\omega_i,\ i=1,2$. If
the spectra of $\omega_i$ are simple (nondegenerate), then
$\ket{\psi} = U_1 U_2 \ket{\phi}$.
\end{theorem}
\begin{proof}
Without loss of generality it can be assumed that bipartition of
$n$ fermionic modes into $m$ and $n-m$ modes is such that $n-m \le
m$ (otherwise the modes can be relabelled).

Consider the state $\omega_2$. In the Hilbert space $H_{n-m}$
choose the orthonormal basis $\{ \ket{e_{j_{m+1}\ldots j_n}} \}$
such that for all $x\in {\mathcal A}_2$
\[
\omega _2(x) = \sum \limits _{j_s=0,1: \ m+1\le s\le n}\lambda
_{j_{m+1}\ldots j_n} \bra{e_{j_{m+1}\ldots j_n}} x
\ket{e_{j_{m+1}\ldots j_n}}.
\]

\noindent Pick the unitary operator $U_2\in {\mathcal A}_2$ such
that
\[
U_2 \ket{j_{m+1}\ldots j_n} = \ket{e_{j_{m+1}\ldots j_n}}, \quad
j_s\in \{0,1\}, \quad m+1 \le s\le n.
\]

\noindent Consider the pure state $\tilde\omega(x) = \omega (U_2 x
U_2^{\dag})$, $x\in{\mathcal A}$. Then for all $x\in\mathcal{A}_2$
we have
\begin{equation}\label{p1}
\tilde \omega _2(x) = \sum \limits_{j_s=0,1: \ m+1 \le s \le n}
\lambda_{j_{m+1} \ldots j_n} \bra{j_{m+1}\ldots j_n} x
\ket{j_{m+1}\ldots j_n}.
\end{equation}

Since $\omega_1$ and $\omega_2$ are equispectral by the statement
of Theorem, the states $\tilde{\omega}_1$ and $\tilde{\omega}_2$
are equispectral too due to the local action of $U_2$.
Consequently, there exists the orthonormal basis $\{
\ket{f_t}\}_{t=0}^{2^{n-m}-1} \cup \{
\ket{g_t}\}_{t=2^{n-m}}^{2^{m}-1}$ in $H_m$ such that for all
$x\in {\mathcal A}_1$
\begin{equation}\label{p2}
\tilde \omega _1(x) = \sum \limits_{t=0}^{2^{n-m}-1}
\lambda_{j_{m+1} \ldots j_n = {\rm bin}(t)} \bra{f_t} x \ket{f_t},
\end{equation}

\noindent where ${\rm bin}(t)$ is the binary representation of
$t$. Pick the unitary operator $U_1 \in \mathcal{A}_1$ such that
\begin{eqnarray}
&& U_1 \, \ket{ \underbrace{0 \ldots 0}_{\footnotesize
\begin{array}{c}
  2m-n \\
  \text{bits} \\
\end{array}} \underbrace{{\rm bin}(t)}_{\footnotesize \begin{array}{c}
  n-m \\
  \text{bits} \\
\end{array}} } =
\ket{f_t}, \quad t = 0, \ldots, 2^{n-m}-1, \nonumber\\
&& U_1 \, \ket{\underbrace{{\rm bin} (t)}_{\footnotesize
\begin{array}{c}
  m \\
  \text{bits} \\
\end{array}}} = \ket{g_t}, \quad t = 2^{n-m}, \ldots, 2^{m}-1. \nonumber
\end{eqnarray}

\noindent It is not hard to see that a vector
\[
\ket{\phi} = \sum \limits_{t=0}^{2^{n-m}-1} \sqrt{
\lambda_{j_{m+1} \ldots j_n = {\rm bin}(t)} }  \, \ket{
\underbrace{0 \ldots 0}_{\footnotesize
\begin{array}{c}
  2m-n \\
  \text{bits} \\
\end{array}} \underbrace{{\rm bin}(t)}_{\footnotesize \begin{array}{c}
  n-m \\
  \text{bits} \\
\end{array}} \underbrace{{\rm bin}(t)}_{\footnotesize \begin{array}{c}
  n-m \\
  \text{bits} \\
\end{array}} }
\]
\noindent generates the state $\Omega (x) = \bra{\phi} x
\ket{\phi}$, $x \in \mathcal{A}$, which satisfies the parity
superselection rule as the number of quasiparticles in $t$-th
summand equals an even number $2\times(\text{\# of 1's in }{\rm
bin(t)})$. Moreover, $\Omega$ has $\tilde{\omega}_2$ and
$\tilde{\tilde{\omega}}_1 (x) = \tilde{\omega}_1 (U_1 x
U_1^{\dag})$ as its partial states. In other words, the
mode-reduced states of $\ket{\psi}\bra{\psi}$ and the mode-reduced
states of $\ket{U_2 U_1 \phi} \bra{U_2 U_1 \phi}$ coincide.

Now suppose that the coincident spectra of $\omega_i$, $i=1,2$ are
simple (nondegenerate). Take the unitary operators $U$ and
$\Gamma$ determined in Eq.~\eqref{u} and consider the state
$\tilde{\Omega}(X) = \Omega(U^{\dag} X U)$, $X\in \mathcal{A}_1
\otimes \mathcal{A}_2$. According to Eq.~\eqref{U-reduction-1},
for all $x\in \mathcal{A}_1$ we have $\tilde{\Omega}_1 (x) =
\tilde{\tilde{\omega}}_1(x)$. Analogously, since $\Omega$ is even,
from Eqs.~\eqref{equality-odd}--\eqref{equality-even} it follows
that for all $x\in \mathcal{A}_2$ the relation $\tilde{\Omega}_2
(x) = \tilde{\omega}_2(x)$ holds.

On the other hand, as the spectrum of $\Omega_{i}$ is simple there
exists the only state on $\mathcal{A}_1 \otimes \mathcal{A}_2$
(namely, $\tilde{\Omega}$) such that its partial traces coincide
with $\Omega_1$ and $\Omega_2$ (see, e.g., \cite{holevo}). Hence,
the same uniqueness holds for the state on $\mathcal{A}$ (namely,
$\Omega$). Thus, there is the only possibility to reconstruct the
entire state from its partial states. It implies that $\omega (x)
= \Omega (U_1^{\dag} U_2^{\dag} x U_2 U_1)$ for all
$x\in\mathcal{A}$.
\end{proof}

Theorem~\ref{theorem-2} characterizes possible states with
equispectral mode-reduced states. Such states do not have to
satisfy the parity superselection rule (cf.
Proposition~\ref{prop-two-mode-case}) but they are obtained from
the superselected states by unitary operators acting on
corresponding parties of the bipartite state.

%%%%%%%%%%%%%%%%%%%%%%%%%%%%%%%%%%%%%%%%%%%%%%%%%%%%%%%%%%%%%%%%%%%%%%%%%%%%%%%%%%%

\section{\label{section-particle-reduction} Spectra of particle-reduced
operators}

In this section, we consider reductions over particles. When the
number of particles $N$ is fixed, as it takes place in quantum
chemistry, the reduction is performed by integrating the density
operator $\varrho(x_1,\ldots,x_N,x_1',\ldots,x_N')$ over some
particles' coordinates~\cite{carlson,coleman,ando,mazziotti}. This
results in the so-called $p$-particle reduced density matrix
($p$-RDM):
\begin{eqnarray}
\label{RDM-usual}
\!\!\!\!\!\!\!\!\! && \varrho_{p\text{-RDM}}(x_1,\ldots,x_p ; x_1',\ldots,x_p') \nonumber\\
\!\!\!\!\!\!\!\!\! && = \int \varrho ( x_1, \ldots, x_p, x_{p+1},
\ldots, x_N ; x_1', \ldots, x_p', x_{p+1}, \ldots, x_N )
\nonumber\\
\!\!\!\!\!\!\!\!\! &&
\qquad\qquad\qquad\qquad\qquad\qquad\qquad\qquad \times dx_{p+1}
\ldots dx_N. \quad
\end{eqnarray}

If a pure fermionic state $\varrho$ has exactly $N$ particles,
then the spectra of reduced density matrices
$\varrho_{p\text{-RDM}}$ and $\varrho_{(N-p)\text{-RDM}}$ are
known to coincide~\cite{carlson}. Eigenvalues of 1-RDM are called
natural occupation numbers. In general case, natural occupation
numbers cannot be arbitrary ones within $[0,1]$ (simple Pauli's
exclusion principle) and must satisfy generalized Pauli
constraints~\cite{borland-dennis,Schilling,benavides-riveros,schilling-2015}.

Let us generalize the construction of reduced density matrices for
states with a varying number of particles. In the second
quantization formalism, the integration in Eq.~\eqref{RDM-usual}
takes the form
\begin{equation}
\label{RDM-second-quantization}
\varrho_{p\text{-RDM}}(s_1,\ldots,s_p ; t_1,\ldots,t_p) = \ave{
a_{s_1}^{\dag} \ldots a_{s_p}^{\dag} a_{t_p} \ldots a_{t_1} },
\end{equation}

\noindent where $s_q, t_q = 1, \ldots, n$, i.e. all modes are
accessible. Note that all $s_q$ are to be different and all $t_q$
are to be different for expression~\eqref{RDM-second-quantization}
not to vanish.

The definition \eqref{RDM-second-quantization} implies that
$p$-RDM does not have unit trace (in contrast to
Eq.~\eqref{RDM-usual}). If the number of particles in the state
$\varrho$ is fixed and equals $N$, then one can introduce a factor
$\frac {(N-p)!}{N!}$ in the right-hand side of
Eq.~\eqref{RDM-second-quantization} to make its trace
unit~\cite{coleman}. As our aim is to deal with general states
that do not have a fixed number of particles, we do not introduce
any factors in the right-hand side of
Eq.~\eqref{RDM-second-quantization}.

As far as the states $\varrho$ with a variable number of particles
are concerned, average values \eqref{RDM-second-quantization} are
still meaningful and can be calculated. This enables one to
construct the particle-reduced density operator in $B(H_n)$ as
follows:
\begin{eqnarray}
\label{RDM-operator} \varrho_p = & \!\!\! \sum\limits_{\scriptsize
\begin{array}{c}
  s_1 < \ldots < s_p \\
  t_1 < \ldots < t_p \\
\end{array}} \!\!\! & \ave{
a_{s_1}^{\dag} \ldots a_{s_p}^{\dag} a_{t_p} \ldots a_{t_1} } \
a_{t_1}^{\dag} \ldots a_{t_p}^{\dag} a_{s_p} \ldots a_{s_1}
\nonumber\\
&& \qquad\qquad \times \prod_{\scriptsize \begin{array}{r}
  s \ne s_1, \ldots, s_p , \\
  t_1, \ldots, t_p \, \\
\end{array}} a_s a_s^{\dag}.
\end{eqnarray}

\noindent Note that the number of particles in the operator
\eqref{RDM-operator} is fixed without regard to a possible
variable number of particles in the original state. Product
$\prod_s$ in formula \eqref{RDM-operator} is responsible for
unoccupied modes.

\begin{example}
\label{example-1-RDM-two-modes} The two-mode fermionic state
$\ket{\psi} = c_{00}\ket{00} + c_{01}\ket{01} + c_{10}\ket{10} +
c_{11} \ket{11}$ has 1-RDM of
the form $\left(%
\begin{array}{cc}
  |c_{10}|^2+|c_{11}|^2 & c_{01} \overline{c_{10}} \\
  \overline{c_{01}} c_{10} & |c_{01}|^2+|c_{11}|^2 \\
\end{array}%
\right)$ whose spectrum is $(|c_{01}|^2 + |c_{10}|^2 +
|c_{11}|^2,|c_{11}|^2) = (1-|c_{00}|^2,|c_{11}|^2)$. There are no
constraints on natural occupation numbers $\lambda_1,\lambda_2 \in
[0,1]$. This is in agreement with a general statement that in Fock
space there are no constraints on the spectra of the 1-RDM other
than those imposed by the Pauli's exclusion
principle~\cite{lopes-2015}. However, if we restrict $\ket{\psi}$
to satisfy the parity superselection rule, then the generalized
constraint on natural occupation numbers appears:
$\lambda_1=\lambda_2$ for states with even parity, $\lambda_2 = 0$
for states with odd parity. \hfill $\square$
\end{example}

\begin{example}
Odd parity 3-mode fermionic state $c_{100}\ket{100} +
c_{010}\ket{010} + c_{001} \ket{001} + c_{111} \ket{111}$ has
natural occupation numbers $\lambda_1 = |c_{100}|^2 + |c_{010}|^2
+ |c_{001}|^2 + |c_{111}|^2 = 1$, $\lambda_2 = \lambda_3 =
|c_{111}|^2$. The requirements $\lambda_1 = 1$ and $\lambda_2 =
\lambda_3$ may be considered as generalized Pauli constraints in
this case. \hfill $\square$
\end{example}

Suppose the physical process of tracing out 1 particle from a
general fermionic state with a variable number of particles. The
result of this operation still contains a variable number of
particles so it cannot be described by a reduced density matrix
\eqref{RDM-second-quantization}. This operation is rather
described by the map
\begin{equation}
\label{Phi-map} \Phi (\varrho) = \sum \limits_j a_j \varrho
a_j^{\dag}.
\end{equation}

\noindent If the state $\varrho$ has a fixed number of particles,
say $N$, then $\Phi (\varrho)$ is an $N-1$ particle operator,
which coincides with the reduced density operator $\varrho_{N-1}$
given by formula~\eqref{RDM-operator} (see the proof below). Thus,
the map $\Phi$ can be considered as a generalization of
integration over one particle coordinates. Sequentially applying
this map $p$ times we get $\Phi^p$, which is nothing else but
tracing out $p$ particles.

\begin{proposition}
\label{prop-N-particle-reductions} Suppose the density operator
$\varrho$ has a fixed number $N$ of particles, then $\Phi^{N-p}
(\varrho) = \varrho_p$.
\end{proposition}
\begin{proof}
According to Corollary~\ref{corollary-1}, $\varrho =
\sum_{|J|=|K|=N} \lambda_{JK} A_{JK}$, where $|J| = \sum_s j_s$ is
the number of particles present in the multiindex $J$. Using the
explicit form of operators~\eqref{c-operators}, we get
\begin{eqnarray}
&& A_{JK} = \underbrace{a_{s_1}^{\dag} \ldots a_{s_N}^{\dag}
a_{t_N} \ldots a_{t_1}}_{\scriptsize \begin{array}{c}
  s_1 < \ldots < s_N: \, j_{s_q} = 1 \\
  t_1 < \ldots < t_N: \, k_{t_r} = 1 \\
\end{array} } \ \prod\limits_{s: \, j_s = k_s = 0} a_s a_s^{\dag}, \\
&& a_l A_{JK} a_l^{\dag} = \left\{ \begin{array}{l}
  a_{s_1}^{\dag} \ldots a_l \ldots a_{s_N}^{\dag} a_{t_N} \ldots a_l^{\dag} \ldots
a_{t_1} \\
\times \prod_{s: \, j_s = k_s = 0} a_s a_s^{\dag} \ \text{~if~} j_l = k_l = 1, \\
  0 \ \text{~otherwise}, \\
\end{array} \right. \qquad
\end{eqnarray}

\noindent which follows from $a_l a_l^{\dag} a_l a_l^{\dag} = a_l
a_l^{\dag}$. On the other hand, substituting $A_{JK}$ for
$\varrho$ in formula~\eqref{RDM-operator} yields
\begin{eqnarray}
(A_{JK})_{N-1} &=& \sum_{l: \, j_l = k_l = 1} a_{s_1}^{\dag}
\ldots a_l \ldots a_{s_N}^{\dag} a_{t_N} \ldots a_l^{\dag} \ldots
a_{t_1}
\nonumber\\
&& \qquad\qquad\qquad\quad \times \prod_{s: \, j_s = k_s = 0} a_s
a_s^{\dag},
\end{eqnarray}

\noindent which implies that $\Phi (\varrho) = \varrho_{N-1}$.
Applying this relation $N-p$ times, we get the statement of
Proposition~\ref{prop-N-particle-reductions}.
\end{proof}

\begin{corollary}
\label{corollary-particle-spectra-coincident} The operators
$\varrho_{p}$ and $\Phi^p(\varrho)$ have the same spectra for pure
states $\varrho$ with a fixed number of particles.
\end{corollary}
\begin{proof}
The statement is an immediate consequence of
Proposition~\ref{prop-N-particle-reductions} and the fact that for
pure states ${\rm Spec}(\varrho_p) = {\rm
Spec}(\varrho_{N-p})$~\cite{carlson}.
\end{proof}

\begin{example}
The spectra of $\varrho_{p}$ and $\Phi^p(\varrho)$ may coincide
not only for states with a fixed number of particles. The
coincidence of spectra of $\varrho_1$ and $\Phi(\varrho)$ takes
place for a general two-mode fermionic state $\ket{\psi} =
c_{00}\ket{00} + c_{01}\ket{01} + c_{10}\ket{10} + c_{11}
\ket{11}$. It is not hard to see from the matrix representation
\eqref{total-density-matrix} of $\Phi(\ket{\psi}\bra{\psi})$ which
reads
\begin{equation}
\label{Phi-acts-on-two-mode-state}
\left(%
\begin{array}{cccc}
  |c_{01}|^2+|c_{10}|^2 & c_{10} \overline{c_{11}} & -c_{01} \overline{c_{11}} & 0 \\
  \overline{c_{10}} c_{11} & |c_{11}|^2 & 0 & 0 \\
  - \overline{c_{01}} c_{11} & 0 & |c_{11}|^2 & 0 \\
  0 & 0 & 0 & 0 \\
\end{array}%
\right).
\end{equation}

\noindent Eigenvalues of \eqref{Phi-acts-on-two-mode-state} are
the same as those of 1-RDM found in
Example~\ref{example-1-RDM-two-modes}. \hfill $\square$
\end{example}

Consideration of other examples (general 3- and 4-mode pure
states) also results in equispectral particle reductions
$\varrho_{p}$ and $\Phi^p(\varrho)$. We can make a conjecture that
spectra of operators $\varrho_{p}$ and $\Phi^p(\varrho)$ coincide
for any pure fermionic state $\varrho$.

%%%%%%%%%%%%%%%%%%%%%%%%%%%%%%%%%%%%%%%%%%%%%%%%%%%%%%%%%%%%%%%%%%%%%%%%%%%%%%%%%%%

\section{\label{section-conclusions} Conclusions}

We have constructed mode- and particle-reduced fermionic density
operators for states with a varying number of particles and
analyzed their spectra. The mode reduction is based on the
restriction functional $\omega(x) = {\rm Tr}(\varrho x)$ to
subalgebra, whereas the particle reduction is realized via two
objects: $p$-particle operator $\varrho_p$ and the result of
tracing out $p$ particles, $\Phi^p(\varrho)$. As a byproduct, we
have analyzed the density matrix formalism for general case of
fermionic states and provided the explicit construction of density
matrices (Corollary~\ref{corollary-matrix-of-averages} and
Example~\ref{example-density-matrices}). The developed formalism
clearly indicates that the conventional partial trace methods are
not applicable to fermionic states.

We have addressed the problem of finding general pure fermionic
states such that their mode-reduced density matrices are
equispectral. Equispectrality automatically takes place for
bosonic systems and systems of distinguishable particles, however,
it does not necessarily hold true for systems composed of
indistinguishable fermionic quasiparticles. On the other hand,
equispectrality is a natural quantum information property that
reflects the fact that entropies of subsystems must coincide.

For mode-reduced states, we have found necessary and sufficient
conditions for their spectra to be identical. The results of
Proposition~\ref{prop-two-mode-case} and Theorem~\ref{theorem-2}
can also be interpreted as the construction of valid purifications
for fermionic density operators. Purification is a frequently used
tool in quantum information theory, and we believe that these
results may turn out to be useful in characterization of fermionic
channels.

For particle-reduced states, we have shown that $\Phi^p(\varrho)$
is a valid reduction, which coincides with $\varrho_{N-p}$ if the
state $\varrho$ has exactly $N$ particles. Basing on the examples,
we have conjectured that spectra of $\varrho_{p}$ and
$\Phi^p(\varrho)$ coincide not only for $N$-particle states but
also for a general fermionic pure state $\varrho$. We have
demonstrated that the natural occupation numbers $\lambda_1 \geq
\ldots \geq \lambda_n$ (spectrum of 1-RDM) must obey generalized
Pauli constraints not only for $N$-particle states but also for
states satisfying the parity superselection rule. For instance,
even parity two-mode states necessarily satisfy
$\lambda_1=\lambda_2$, odd parity two-mode states have $\lambda_2
= 0$, and odd parity three-mode states fulfil the requirements
$\lambda_1 = 1$, $\lambda_2 = \lambda_3$. The generalization of
these constraints is a problem for a future research.

%%%%%%%%%%%%%%%%%%%%%%%%%%%%%%%%%%%%%%%%%%%%%%%%%%%%%%%%%%%%%%%%%%%%%%%%%%%%%%%%%%%%

\begin{acknowledgements}
The authors are delighted to thank A.S. Holevo for a motivation of
this work and illuminating discussions. The authors are grateful
to Christian Schilling for fruitful discussions and bringing
Ref.~\cite{lopes-2015} to our attention.
Propositions~\ref{prop-even} and \ref{prop-matrix-unit},
Corollary~\ref{corollary-1}, Theorems~\ref{theorem-1} and
\ref{theorem-2} are proved by G.G. Amosov.
Corollary~\ref{corollary-matrix-of-averages},
Proposition~\ref{prop-two-mode-case},
Examples~\ref{example-density-matrices} and \ref{example-spectra},
Section~\ref{section-particle-reduction} are due to S.N. Filippov.
Both authors discussed the results and commented on the
manuscript. The work of G.G. Amosov is supported by Russian
Science Foundation under grant No. 14-21-00162 and performed in
Steklov Mathematical Institute of Russian Academy of Sciences.
S.N. Filippov's work on Sections 3 and 4 is supported by the
Russian Foundation for Basic Research under Project No.
16-37-60070 mol\_a\_dk and performed in Institute of Physics and
Technology, Russian Academy of Sciences. S.N. Filippov's work on
Sections 5 is supported by Russian Science Foundation under
project No. 16-11-00084 and performed in Moscow Institute of
Physics and Technology.
\end{acknowledgements}

\end{document}